\documentclass[12pt]{colt2017} 

\newcommand{\Gc}{\mathcal{G}}
\newcommand{\Fc}{\mathcal{F}}
\newcommand{\Bc}{\mathcal{B}_f}

\renewcommand{\P}{\mathsf{P}}
\newcommand{\E}{\mathsf{E}}
\newcommand{\IER}{\mathsf{IER}}
\newcommand{\Geom}{\mathsf{Geom}}

\newcommand{\dg}{\textup{deg}}

\newtheorem{claim}[theorem]{Claim}
\newtheorem{assumption}{Assumption}
\newtheorem*{problem}{Problem}

\title[Two-Sample Tests for Random Graphs]{Two-Sample Tests for Large Random Graphs Using Network Statistics}
\usepackage{times}

 \coltauthor{\Name{Debarghya Ghoshdastidar} \Email{debarghya.ghoshdastidar@uni-tuebingen.de}\\
 \addr Department of Computer Science, University of T{\"u}bingen,   \\
          Sand 14, 72076 T{\"u}bingen, Germany
 \AND
 \Name{Maurilio Gutzeit} \Email{mgutzeit@uni-potsdam.de}\\
 \Name{Alexandra Carpentier} \Email{carpentier@uni-potsdam.de}\\
 \addr Department of Mathematics, University of Potsdam, \\
 Karl-Liebknecht-Strasse 24-25,  D-14476 Potsdam OT Golm, Germany
 \AND
 \Name{Ulrike von Luxburg} \Email{luxburg@informatik.uni-tuebingen.de}\\
 \addr Department of Computer Science, University of T{\"u}bingen,   \\
          Sand 14, 72076 T{\"u}bingen, Germany
 }

\begin{document}

\maketitle

\begin{abstract}
We consider a two-sample hypothesis testing problem, where the distributions are defined on the space of undirected graphs, 
and one has access to only one observation from each model.
A motivating example for this problem is comparing the friendship networks on Facebook and LinkedIn.
The practical approach to such problems is to compare the networks based on certain network statistics.
In this paper, we present a general principle for two-sample hypothesis testing in such scenarios without making any assumption about the network generation process.
The main contribution of the paper is a general formulation of the problem based on concentration of network statistics, and consequently, a consistent two-sample test that arises as the natural solution for this problem.
We also show that the proposed test is minimax optimal for certain network statistics. 
\end{abstract}


\section{Introduction}

Hypothesis testing has received considerable attention in recent times, particularly in the high-dimensional setting, where the number of observations is much smaller than the dimension of data~\citep{Gretton_2012_jour_JMLR,Mukherjee_2015_jour_AnnStat}.
From this perspective, an extreme situation arises in network analysis, where we observe one large network and need to draw inference based on a single instance of a high-dimensional object.
Surprisingly, simple tests are known to work quite well in many of these graph inference problems. 
For instance, \citet{Bubeck_2016_jour_RSA} compute the number of signed triangles in a graph to test whether the graph has an underlying geometric structure, or it is purely random (Erd{\"o}s-R{\'e}nyi).
In the context of community detection, \citet{AriasCastro_2014_jour_AnnStat} use modularity and clique number to detect the presence of dense sub-graphs.
Along similar lines, \citet{Lei_2016_jour_AnnStat} use the largest few singular values of the adjacency matrix to decide whether a graph has $k$ or more communities. 
In all these works, the constructed tests exploit the fact that there exist some network statistics, for example number of triangles or the graph spectrum, that help to distinguish between the null and the alternative hypotheses.
This principle is not restricted to specific testing problems, but a general approach often used by practitioners from different disciplines in the context of testing and fitting of network models~\citep{Rubinov_2010_jour_NeuroImage}.
For instance,  \citet{Klimm_2014_jour_PLoSComputBiol} write that the Erd{\"o}s-R{\'e}nyi (ER) model is a poor fit for the brain network since \emph{``the clustering coefficient (in ER) is smaller and the path length shorter than that of anatomical brain networks''}.
However, such a statement is purely qualitative, and arises due to the lack of a formal framework for testing random graphs.

The purpose of the present paper is to formulate a general approach for hypothesis testing for large random graphs, particularly when one does not have access to independent and identically distributed samples generated from the same network model.  
We focus on the problem of two-sample hypothesis testing, where one observes two random graphs, possibly of different sizes. Based on the two given graphs, the problem is to decide whether the underlying distributions that generate the graphs are same or different.
The problem surfaces naturally in several domains, for instance, comparison of brain networks of people with and without Alzheimer~\citep{Stam_2007_jour_CerebralCortex}.
However, the problem has been previously studied only in a very restrictive setting, where the two graphs are assumed to be random dot product graphs (RDPG) defined on the same set of vertices~\citep{Tang_2016_jour_JCompGraphStat}. 
The RDPG model has a semi-parametric characterisation, which allows one to estimate the parameters up to some transformation, and then use these estimates for a two-sample test. Obviously, this approach cannot be directly extended to more general classes of network models.  
A more critical limitation is that the study cannot be easily extended to compare graphs that do not have a vertex correspondence, and more generally, graphs of different sizes. 

One can easily see that the aforementioned hypothesis testing problem is ill-posed in general since the distributions that generate graphs of different size reside in different spaces, and hence, cannot be directly compared.   
To resolve this issue, we take the practical approach of comparing two models in terms of network statistics. 
In other words, we compute a function $f$ for both graphs, and decide whether the two graphs are similar or different with respect to $f$.
It is obvious that this approach cannot distinguish between two distributions for which $f$ behave in the same way, but this is a cost that is often incurred in practice, where one needs to know the interesting features that distinguish two network models~\citep{Stam_2007_jour_CerebralCortex,Klimm_2014_jour_PLoSComputBiol}.
Special cases of such a problem can be posed, and even solved, quite easily. For example, a simple situation arises if we restrict only to ER graphs, and $f$ corresponds to the edge probability.
However, it is not clear how one can pose the problem formally in a more general setting, where one does not fix the network statistic or make any assumption on the network model. 
In this paper, we tackle this issue by restricting the class of network statistics $f$ that can be used, and then tailor the hypothesis testing problem to the pre-specified network statistic.
We restrict the class of statistics $f$ to those that concentrate for large graphs, and subsequently pose the testing problem in terms of the point of concentration of $f$.
As a consequence, we also find an uniformly consistent test that arises naturally for this problem, and is also optimal in some cases.

The paper is organised as follows. In Section~\ref{sec_statistic}, we present a general assumption on network statistics, and show that two common statistics, based on (i) number of triangles, and (ii) largest $k$ singular values of adjacency matrix, satisfy this assumption.
We set up the hypothesis testing problem in Section~\ref{sec_problem}, and present a generic two-sample test for the problem. 
We also show that if the null and alternative hypotheses are separated by a certain factor, then the test can achieve arbitrarily small error rates for large enough graphs (see Theorem~\ref{thm_main}). 
We provide more concrete statements in Section~\ref{sec_examples}, where we restrict the discussion to the above mentioned network statistics, and consider specific network models, namely,
edge-independent but inhomogeneous random graphs~\citep{Bollobas_2007_jour_RSA}, and random geometric graphs~\citep{Penrose_2003_book_Oxford}.
Subsequently, in Section~\ref{sec_minimax}, we consider inhomogeneous random graphs and the mentioned network statistics, and prove that the separation condition derived in Theorem~\ref{thm_main} matches with the detection boundary. In other words, the proposed test achieves the minimax rate of testing for certain class of random graphs and network statistics. 
We present proofs of all results in Section~\ref{sec_proof}.
Finally, we conclude in Section~\ref{sec_conclusion} with some discussion and open questions.

\section{Network statistic}
\label{sec_statistic}

We denote the set of all undirected unweighted graphs on $n$ vertices by $\Gc_n$, and
let $\Fc_n$ be the set of all probability distributions on $\mathcal{G}_n$.
For convenience, we write $\mathcal{G}_{\geq N}$ instead of $\bigcup_{n\geq N} \mathcal{G}_n$.
If $Q$ is a distribution, we use $G\sim Q$ to say that $G$ is generated from $Q$. 
We also use the notation $o_n(1)$ to denote the class of functions of $n$ that vanish as $n\to\infty$.

Let $(\mathcal{S},d)$ be a metric space. We call any mapping $f: \mathcal{G}_{\geq1} \to \mathcal{S}$ a network statistic. 
Observe that the function $f$ acts on all possible graphs, and maps them to a metric space.
Some natural examples include scalar functions such as edge density or number of triangles; multivariate statistics like largest $k$ singular values of adjacency matrix; and even quantities like degree distribution. For the last example, $\mathcal{S}$ corresponds to the space of all distributions on natural numbers endowed with a suitable metric.
All the popular network measures~\citep[see][]{Rubinov_2010_jour_NeuroImage,Bounova_2012_jour_PhysRevE} can be put in this framework.
However, not all such functions are useful for hypothesis testing. The following assumption provides a way to characterise the useful network statistics.

\begin{assumption}[Statistic that concentrates for large graphs]
\label{assum}
Given a network statistic $f$, we assume that there exist
\begin{itemize}
\item a set $\Fc_{f} \subset \Fc_{\geq1}$ (a class of distributions that meet the requirements for concentration of $f$), 
\item a mapping $\mu:\Fc_{f}\to\mathcal{S}$ (the point of concentration of $f$ for a particular distribution), and 
\item a mapping $\sigma: \Fc_{f}\to\mathbb{R}$ (the deviation of $f$ from $\mu$),  
\end{itemize}
such that 
\begin{equation}
\sup_{Q\in\Fc_{f}\cap\Fc_n} \P_{G\sim Q} \big( d(f(G),\mu(Q)) > \sigma(Q) \big) \in o_n(1) \;.
\label{eqn_concen}
\end{equation}
\end{assumption}

The above assumption implies that if we consider a {``good''} network statistic $f$, then for any distribution $Q\in\Fc_n$ and $G\sim Q$, the computed statistic $f(G)$ concentrates about an element in $\mathcal{S}$. 
The function $\mu$ maps the distribution $Q$ to this point of concentration. A natural choice is $\mu(Q) = \E_{G\sim Q}[f(G)]$ if the expectation exists, where the quantity $\sigma$ would be related to the standard deviation of $f(G)$.
However, we show later that a more general definition for $\mu$ and $\sigma$ helps to formulate some common testing problems.
To this end, one may note that the concentration~\eqref{eqn_concen} may not occur for all distributions in $\Fc_{\geq1}$, a typical example being models that are very sparse. 
Hence, Assumption~\ref{assum} restricts the class of distributions to $\Fc_f$ that can be viewed as a subset of distributions for which concentration occurs.

We note that Assumption~\ref{assum} is quite weak in general since it allows several trivial cases. For example, if $f$ corresponds to the average probability of edges, then setting the deviation $\sigma(\cdot)=1$ shows that $f$ satisfies Assumption~\ref{assum} for any arbitrary $\mu:\Fc_{\geq1}\to[0,1]$.  Hence, in this case, satisfying Assumption~\ref{assum} is of no practical significance.
We now consider few common statistics to show that the assumption often has interesting and useful consequences.

\begin{example}[Average probability of triangle]
Let $G$ be an undirected graph on $n$ vertices with adjacency matrix $A_G$, then the univariate statistic
\begin{equation}
f_\Delta(G) = \frac{1}{\binom{n}{3}} \sum_{i<j<k} (A_G)_{ij} (A_G)_{jk} (A_G)_{ik}
\end{equation}
provides an estimate of the average probability of occurrence of a triangle. 
Note that $f_\Delta$ maps $\Gc_{\geq1}$ to $\mathcal{S}=[0,1]$ with metric $d$ being the absolute difference. 
\end{example}

The above function $f_\Delta$ is a normalised version of the number of triangles, where the normalisation makes the statistic independent of the graph size. 
For instance, any Erd{\"o}s-R{\'e}nyi (ER) graph $G$ with edge probability $p$ satisfies $\E_G[f_\Delta(G)] = p^3$ irrespective of the graph size. The following result provides a choice of $\sigma$ such that $f_\Delta$ satisfies Assumption~\ref{assum} for a very broad class of distributions.

\begin{lemma}[$f_\Delta$ satisfies Assumption~\ref{assum} under a limited correlation condition]
\label{lem_fDelta_assum}
Define the quantity $\mu(Q) = \E_{G\sim Q} [f_\Delta(G)]$, and
let $\Fc_f$ be all distributions on graphs such that the presence of any triangle is not correlated with indicators of any non-overlapping edge or triangle, that is,
\begin{align*}
\Fc_f = \big\{ Q \in \Fc_{\geq1} : \text{for } &G\sim Q, \text{ and any } i<j<k \text{ and } i'<j'<k' \text{ with }
 |\{i,j,k\}\cap\{i',j',k'\}|\leq1,
\\&\E_{G\sim Q}[  \Delta_{ijk} \Delta_{i'j'k'}]= \E_{G\sim Q}[ \Delta_{ijk}]\E_{G\sim Q}[ \Delta_{i'j'k'}], \text{ and}
\\&\E_{G\sim Q}[ \Delta_{ijk} (A_G)_{i'j'} ]= \E_{G\sim Q}[ \Delta_{ijk}]\E_{G\sim Q}[ (A_G)_{i'j'} ] \big\}  ,
\end{align*} 
where we use the notation $\Delta_{ijk} = (A_G)_{ij} (A_G)_{jk} (A_G)_{ik}$.
Then $f_\Delta$ satisfies Assumption~\ref{assum} for above $\mu$ and $\Fc_f$ with $\sigma(Q) =   \sqrt{(3D_Q+1)\mu(Q)\ln n/\binom{n}{3}}$, where $n$ is the size of graph generated from $Q$ and $D_Q$ is the maximum expected degree of any node.
\end{lemma}

The above class of distributions encompasses a wide range of real-world network models since we only require that any triangle is uncorrelated from any non-overlapping edge or triangle. This result obviously holds for graphs with independent edges, but also other models such as certain random geometric  graphs.
For instance, \citet{Bubeck_2016_jour_RSA} use the following definition for geometric graphs. One samples $n$ random vectors $x_1,\ldots,x_n$ i.i.d. uniform from the unit ball in $\mathbb{R}^r$, which correspond to the vertices of the graph. For a given $p\in[0,1]$, edge $(i,j)$ is added if $x_i^Tx_j \geq \tau_{r,p}$, where the threshold $\tau_{r,p}$ is set such that $\P(x_i^Tx_j \geq \tau_{r,p} ) = p$.
We represent this class of distributions by $\Geom$, and note that any $Q\in\Geom\cap\Fc_n$ is defined by two parameters: the dimension of underlying Euclidean space, $r_Q$, and the edge probability, $p_Q$.
Due to Lemma~\ref{lem_fDelta_assum}, we can say that the statistic $f_\Delta$ satisfies Assumption~\ref{assum} when we consider distributions from $\Geom$ class.
\begin{lemma}[$f_\Delta$ satisfies Assumption~\ref{assum} for $\Geom$ class]
\label{lem_fDelta_Geom_assum}
For any $Q\in \Geom\cap\Fc_n$ with parameters $r_Q, p_Q$, define $\mu(Q) = \E_{G\sim Q} [f_\Delta(G)]$. For any absolute constant $C>0$,  $f_\Delta$ satisfies Assumption~\ref{assum} for the choice of functions  $\sigma(Q) = \sqrt{(3np_Q+1)\mu(Q)\ln n/\binom{n}{3}}$, and
\begin{equation*}
\Fc_f = \bigcup_{n\geq1} \left\{Q \in \Geom\cap\Fc_n : p_Q \geq \frac{1}{n} \text{ and } C \leq r_Q \leq \frac{(np)^4(\ln \frac{1}{p_Q})^3}{(\ln n)}  \right\} \;. 
\end{equation*} 
\end{lemma}

We note that the conditions in $\Fc_f$  stated above are not necessary at this stage of discussion. However, our subsequent discussion on hypothesis testing problem and approach require an estimate of $\sigma$ from the random graph, which in turn imposes few additional constraints on the model. In order to simplify our later discussions based on this setting, we restrict to the smaller set $\Fc_f$ defined above.  

It is easy to see that the choice of $\sigma$ mentioned in Lemmas~\ref{lem_fDelta_assum} and~\ref{lem_fDelta_Geom_assum} is not unique, and using a larger deviation function does not lead to violation of the assumption. However, in some cases, one can even consider a smaller deviation function provided that $\Fc_f$ is restricted accordingly. 
Consider the class of inhomogeneous random graphs with independent edges~\citep{Bollobas_2007_jour_RSA}, and denote the set of corresponding distributions by $\IER$.
In such graphs, each edge occurs with a different probability, and hence, any $Q\in\IER\cap\Fc_n$ is characterised by a $n\times n$ symmetric matrix $M_Q$ such that $\E_{G\sim Q}[A_G] = M_Q$. 
Hence, in this case, $\mu(Q) = \E_{G\sim Q} [f_\Delta(G)] = \frac{1}{\binom{n}{3}} \sum\limits_{i<j<k} (M_Q)_{ij} (M_Q)_{jk} (M_Q)_{ik}$.
If $\Fc_f = \IER$, then Lemma~\ref{lem_fDelta_assum} provides a choice of $\sigma$. The following result shows that if we only consider sparse $\IER$ graphs, then $f_\Delta$ also satisfies the same assumption for a smaller deviation function. 

\begin{lemma}[$f_\Delta$ satisfies Assumption~\ref{assum} for semi-sparse $\IER$ class]
\label{lem_fDelta_IER_assum}
For any $Q\in \IER\cap\Fc_n$ with associated matrix $M_Q$, if $\mu(Q) = \E_{G\sim Q} [f_\Delta(G)]  = \frac{1}{\binom{n}{3}} \sum\limits_{i<j<k} (M_Q)_{ij} (M_Q)_{jk} (M_Q)_{ik}$, then 
$f_\Delta$ satisfies Assumption~\ref{assum} for the choices $\sigma(Q) = 2\sqrt{\mu(Q)\ln n/\binom{n}{3}}$, and
\begin{equation*}
\Fc_f = \bigcup_{n\geq1} \left\{Q \in \IER\cap\Fc_n : \mu(Q) \geq \frac{\ln n}{n^3} \text{ and } \max_{i,j} \sum_{k\neq i,j} (M_Q)_{ik}(M_Q)_{jk} \leq 1 \right\} \;. 
\end{equation*} 
\end{lemma}
Observe that the condition $\sum_k (M_Q)_{ik}(M_Q)_{jk} \leq 1$ is equivalent to stating that the graphs are sparse enough so that, in the expected sense, no edge appears in more than one triangle. 
The condition on minimum growth rate of $\mu(Q)$ is not necessary at this stage, and simply ensures that $\mu(Q)$ can be estimated from $G\sim Q$. 

Based on Lemma~\ref{lem_fDelta_assum}, one may also consider a combination of both $\IER$ and $\Geom$ classes, which allows one to tackle problems where two random graphs are chosen from $\IER\cup\Geom$. 
This emphasises the flexibility of the present discussion and the forthcoming results in the sense that as long as one can show that a network statistic concentrates, one can apply the framework and result of this paper.
We now look at one more common statistic.

\begin{example}[Normalised largest $k$ singular values] 
Let $\lambda_1(A_G)\geq\ldots\geq\lambda_k(A_G)$ be the largest $k$ singular values of adjacency matrix $A_G$. Then the multivariate statistic
\begin{equation}
f_{\lambda}(G) = \frac1n \big(\lambda_1(A_G),\lambda_2(A_G),\ldots,\lambda_k(A_G)\big)^T\;,
\end{equation}
maps every graph to $\mathcal{S}=\mathbb{R}^k$ endowed with any standard metric $d$. For concreteness, we assume $d$ is the Euclidean distance.
\end{example}
This is yet another statistic whose concentration properties have been well studied, especially for the $\IER$ class of random graphs.
The scaling of $\frac1n$ again helps to reduce the dependence on graph size. Alternative ways to achieve this could be to consider the spectrum of normalised adjacency or normalised Laplacian~\citep{Chung_1997_book_AMS}.

In the case of $\IER$ graphs, \citet{Alon_2002_jour_IsraelJMath} provide the rate of concentration of eigenvalues about their expected values. Based on this, one may set $\mu(Q) = \E_{G\sim Q}[f_\lambda(G)]$, and claim that $f_\lambda$ satisfies Assumption~\ref{assum}.
But a more interesting fact, from a practical perspective, is that for $Q\in\IER\cap\Fc_n$, the statistic $f_\lambda(G)$ also concentrates about $\frac1n \big(\lambda_1(M_Q),\lambda_2(M_Q),\ldots,\lambda_k(M_Q)\big)^T$.
Thus, one can make the following claim for $f_\lambda$ using concentration results of~\citet{Lu_2013_jour_EJComb}.

\begin{lemma}[$f_\lambda$ satisfies Assumption~\ref{assum} for semi-sparse $\IER$ class]
\label{lem_flambda_assum}
For any $Q\in \IER\cap\Fc_n$ with associated matrix $M_Q$, define $\mu(Q) = \frac1n \big(\lambda_1(M_Q),\ldots,\lambda_k(M_Q)\big)^T$, and let $D_Q =\max_i  \sum_j (M_Q)_{ij}$. Then 
$f_\lambda$ satisfies Assumption~\ref{assum} for the choices $\sigma(Q) = \frac{2.1}{n}\sqrt{kD_Q}$, and
\begin{equation*}
\Fc_f = \bigcup_{n\geq1} \left\{Q \in \IER\cap\Fc_n : D_Q \geq (\ln n)^{4.1} \right\} \;, 
\end{equation*} 
where the constant 2.1 (or 4.1) may be replaced by any value greater than 2 (resp., 4).
\end{lemma}

The restricted set $\Fc_f$ is crucial here. One can show that $f_\lambda$ does not satisfy Assumption~\ref{assum} over the larger set $\IER\cap\Fc_{\geq1}$.  
In particular, set $k=1$ and consider any ER distribution $Q\in\Fc_n$ with edge probability $\frac{c}{n}$ for some constant $c>0$. Then it is known that $f_\lambda(G) \geq \frac1n\sqrt{\frac{\ln n}{\ln \ln n}}$ for any $G\sim Q$~\citep[see][]{Krivelevich_2003_jour_CombinProbabComput}, whereas $\mu(Q)\leq \frac{c}{n}$ decays much faster.
This shows the importance of restricting the statement of Assumption~\ref{assum} to a suitable set $\Fc_f$.
Moreover, we note that in the case of $f_\lambda$, one may use an alternative choice of $\sigma$ that provides concentration even for sparser graphs $D_Q \geq (\ln n)^{1.1}$~\citep{Lei_2015_jour_AnnStat}. Such a choice of $\sigma$ depends on $\max_{ij} (M_Q)_{ij}$, which is difficult to estimate from a single random graph.

\section{Two-sample hypothesis testing}
\label{sec_problem}

Based on the discussions about the network statistic, we are now prepared for a formal statement of the hypothesis testing problem under consideration.
\begin{problem}
Let $f$ be a pre-specified network statistic that satisfies Assumption~\ref{assum} with the associated quantities $\Fc_{f}$ and $\mu$.
Let $Q,Q'\in\Fc_{f}$, and $G$ and $G'$ be random graphs (of possibly different sizes) generated from $Q$ and $Q'$, respectively.
Given $G$ and $G'$, we test the null hypothesis
\begin{equation*}
H_0 : Q,Q' \in \Fc_f \text{ with } d(\mu(Q), \mu(Q'))  \leq \epsilon(Q,Q')
\end{equation*}
against the alternative hypothesis 
\begin{equation*}
H_1 : Q,Q' \in \Fc_f \text{ with } d(\mu(Q), \mu(Q'))  > \rho(Q,Q'), 
\label{eqn_alterhyp}
\end{equation*}
where $\epsilon,\rho$ are non-negative scalar functions of the distributions $Q,Q'$ so that $\epsilon(Q,Q') \leq \rho(Q,Q')$. 
\end{problem}

Setting $\epsilon(\cdot,\cdot)=0$ restricts $H_0$ to the case where $\mu(Q)=\mu(Q')$. However, as shown later, it is often useful to provide some leeway by allowing $\epsilon$ to be positive. 
On the other hand, the function $\rho$ plays the role of a separation that is often required in hypothesis testing~\citep[see, for instance,][]{Ingster_2000_jour_ESAIM}.
In the present context, we later show that if $\rho(Q,Q')$ is too small, then there exist hypotheses that cannot be distinguished by any test.

In the rest of the section, we construct a two-sample test, and show that for certain range of $\epsilon$ and $\rho$, the test can achieve arbitrarily small error for large graphs.
In Section~\ref{sec_minimax}, we use examples to demonstrate that the test is minimax optimal in some cases. Before proceeding, we need few more definitions.
For any statistic $f$ and distribution $Q\in \Fc_f$, we define
$\Bc(Q,\epsilon) = \big\{ Q'\in\Fc_f : d(\mu(Q),\mu(Q')) \leq \epsilon(Q,Q') \big\}$.
Intuitively, one can think of this set as the inverse image of a closed ball in $\mathcal{S}$ centred at $\mu(Q)$ and radius specified by the function $\epsilon$.
Similarly, we define a complement of the ``ball'' for the function $\rho$ as
$\overline{\Bc(Q,\rho)} = \big\{ Q'\in\Fc_f : d(\mu(Q),\mu(Q')) > \rho(Q,Q') \big\}$.
In this terminology, one can see that $H_0$ is true if $Q'\in \Bc(Q,\epsilon)$, while $H_1$ is true for $Q'\in\overline{\Bc(Q,\rho)}$.

\subsection{A consistent two-sample test}

We now construct a test based on concentration of $f$ that is ensured by Assumption~\ref{assum}.
To construct the test, we require an additional assumption on $f$, namely, the fact that $\sigma$ as defined in Assumption~\ref{assum} can be estimated accurately from the graph.
\begin{assumption}[$\sigma(Q)$ can be estimated from $G$]
\label{assum_est}
Let $\Fc_f$ and $\sigma$ be as defined in Assumption~\ref{assum}. There exists a function $\widehat\sigma : \mathcal{G}_{\geq1} \to \mathbb{R}$ with the property that for any $\delta>0$,
\begin{equation}
\sup_{Q\in\Fc_{f}\cap\Fc_n} \P_{G\sim Q} \big( \vert\widehat\sigma(G) - \sigma(Q) \vert > \delta\cdot\sigma(Q) \big)  \in o_n(1).
\end{equation}
\end{assumption}

This assumption rules out some possibilities. For example, one could also state Lemma~\ref{lem_flambda_assum} with $\sigma$ defined in terms of maximum edge probability, $\max_{ij} (M_Q)_{ij}$ instead of maximum expected degree $D_Q$.
However, in that case, it would be impossible to estimate $\sigma$ for a single random graph.
Thus, such a concentration of $f_\lambda$ stated in terms of $\max_{ij} (M_Q)_{ij}$ is not useful for the test described below. 

The two-sample test that we propose is quite straightforward. Given the random graphs $G$ and $G'$, we define the test statistic
\begin{equation}
\label{eqn_teststat}
T(G,G') = \frac{d(f(G), f(G'))}{2\widehat\sigma(G) + 2\widehat\sigma(G')} \;,
\end{equation}
Based on the test statistic in~\eqref{eqn_teststat}, we accept the null hypothesis $H_0$ if $T(G,G') \leq1$, and reject it if $T(G,G')>1$. 
The following result shows  that the above test is uniformly consistent for a suitable network measure $f$, and a large enough separation $\rho$. 
\begin{theorem}[Proposed test is uniformly consistent]
\label{thm_main}
Let $f$ be a network statistic satisfying Assumptions~\ref{assum} and~\ref{assum_est}, and $\Fc_f$ and $\sigma$ be as specified in the assumptions. 
If the functions $\epsilon$ and $\rho$ satisfy
\begin{equation}
\epsilon(Q,Q') \leq 0.5(\sigma(Q) + \sigma(Q'))
\qquad \text{and} \qquad
\rho(Q,Q') \geq 3.5(\sigma(Q) + \sigma(Q'))
\label{eqn_separation_condns}
\end{equation}
for all $Q,Q'\in\Fc_{f}$, 
then 
\begin{align*}
\sup_{Q\in\Fc_{f} \cap \Fc_{\geq n}} \bigg( 
\sup_{Q'\in\Bc(Q,\epsilon)\cap \Fc_{\geq n}} &\underbrace{\P_{G\sim Q,G'\sim Q'}(T(G,G') >1)}_{\textup{Type-I error}} + 
\\& \qquad
\sup_{Q'\in\overline{\Bc(Q,\rho)}\cap \Fc_{\geq n}}\underbrace{\P_{G\sim Q,G'\sim Q'}(T(G,G') \leq1)}_{\textup{Type-II error}} \bigg)
\in o_n(1).
\label{eqn_alphatest_thm}
\end{align*}
\end{theorem}

In simple terms, the above theorem states that if the hypothesis testing problem is defined with $\epsilon$ and $\rho$ satisfying~\eqref{eqn_separation_condns}, and both graphs have at least $n$ vertices, then for any pair of $Q,Q'$, the two-sample test described above achieves an error rate (Type-I $+$ Type-II) that vanishes as $n\to\infty$.
Thus, the test is an uniformly consistent test. We note here that unlike the standard literature on hypothesis testing, we consider the asymptotics in the size of the graphs rather than the number of independent observations.

One may use a weaker variant of Assumption~\ref{assum_est} , where one needs to estimate some upper bound of $\sigma$. This results in a weaker test that can only distinguish models with a larger separation. 

\section{Examples}
\label{sec_examples}

We elaborate our discussion on Theorem~\ref{thm_main} using the two statistics $f_\Delta$ and $f_\lambda$.

\subsection{Two-sample testing using $f_\Delta$}
Recall that $f_\Delta$ is defined as an estimate of the mean probability of triangle occurrence, and for both $\IER$ and $\Geom$ classes, $f_\Delta(G)$ concentrates at $\mu(Q) = \E_{G\sim Q}[f_\Delta(G)]$.
Since, $f_\Delta$ satisfies Assumption~\ref{assum}, one can pose a two-sample testing problem as described in Section~\ref{sec_problem}, where the set of distributions need to be restricted to a suitable $\Fc_f$.
Let us also simplify the problem by fixing $\epsilon(Q,Q') = 0$, that is, we do not distinguish between two models if they are mapped into the same point by $\mu$.
Based on Theorem~\ref{thm_main}, we claim that when the problem is restricted to semi-sparse $\IER$ class, then the statistic in~\eqref{eqn_teststat} leads to a consistent test for appropriately specified $\rho$.

\begin{corollary}[Consistent test for semi-sparse $\IER$ using $f_\Delta$]
\label{cor_fDelta_IER}
Consider the setting of Lemma~\ref{lem_fDelta_IER_assum}. Then $f_\Delta$ satisfies Assumption~\ref{assum_est} for $\widehat\sigma(G) = 2\sqrt{f_\Delta(G)\ln n/\binom{n}{3}}$, where $n$ is number of vertices in $G$. 
Furthermore, the proposed test is uniformly consistent for the testing problem with 
\begin{equation*}
 \rho(Q,Q') \geq 7\left( \sqrt{\frac{\mu(Q)\ln n}{\binom{n}{3}}} + \sqrt{\frac{\mu(Q') \ln n'}{\binom{n'}{3}}}\right)
\end{equation*} 
where $Q\in\Fc_f\cap\Fc_n$ and $Q'\in\Fc_f\cap\Fc_{n'}$.
\end{corollary}

The above result provides a good estimator for the deviation function $\sigma$, and the rest follows immediately from Theorem~\ref{thm_main}. The sufficient condition on $\rho$ stated above may not seem very intuitive, but we show in the next section that the condition is also necessary (up to logarithmic factor) for two-sample testing of $\IER$ graphs using $f_\Delta$.  
In the case of the $\Geom$ class, we obtain a result similar to Corollary~\ref{cor_fDelta_IER}.

\begin{corollary}[Consistent test for $\Geom$ using $f_\Delta$]
\label{cor_fDelta_Geom}
Consider the setting of Lemma~\ref{lem_fDelta_Geom_assum}. For any graph $G$ with $n$ vertices, let $\widehat{p}_G$ denote the estimated edge density, that is, $\widehat{p}(G) = {\binom{n}{2}}^{-1}\sum_{i<j} (A_G)_{ij}$.
Then $f_\Delta$ satisfies Assumption~\ref{assum_est} for $\widehat\sigma(G) = \sqrt{(3n\widehat{p}(G)+1)f_\Delta(G)\ln n/\binom{n}{3}}$. 

As a consequence, the proposed test is uniformly consistent for the testing problem with 
\begin{equation*}
 \rho(Q,Q') \geq {3.5}\left( \sqrt{\frac{(3np_Q+1)\mu(Q)\ln n}{\binom{n}{3}}} + \sqrt{\frac{(3n'p_{Q'}+1)\mu(Q')\ln n'}{\binom{n'}{3}}}\right)
\end{equation*} 
where $Q\in\Fc_f\cap\Fc_n$ and $Q'\in\Fc_f\cap\Fc_{n'}$.
\end{corollary}

As mentioned before, one can also pose a testing problem on $\IER\cup\Geom$. In particular, we consider the following two-sample version of the problem studied in~\citep{Bubeck_2016_jour_RSA}.
Let there be a sequence of probabilities $(p_n)_{n\geq 1}\subset[0,1]$, and a sequence of dimensions $(r_n)_{n\geq1}\subset\mathbb{N}$. 
For each $n$, consider a set containing exactly two distributions from $\Fc_n$: (i) an ER distribution with edge probability $p_n$, and (ii) the other from $\Geom$ class with parameters $p_n$ and $r_n$.
Let two graphs be generated on $n$ vertices from either of these two distributions. The problem is to test whether both are generated from the same model, or different models.

\begin{corollary}[Consistent test for distinguishing between ER and $\Geom$]
\label{cor_fDelta_ERvGeom}
Consider above problem with $\frac{\ln n}{n}\leq p_n \leq\frac{1}{\sqrt{n}}$ for all $n$. The proposed test is uniformly consistent if $r_n = o_n\left( (\ln \frac{1}{p_n})^3 \right)$.
\end{corollary} 

The above result simply implies that the condition for identifiability derived by \citet{Bubeck_2016_jour_RSA} for the sparse one-sample version of the problem remains unchanged in the two-sample case even if we do not assume any knowledge about the parameters. 
The restriction on $p_n$ is a consequence of Lemma~\ref{lem_fDelta_IER_assum}, and the lower bound also helps in accurate estimation of the unknown $p_n$.
It is also easy to verify that the above result holds even when the two graphs are of different size but the edge probabilities are same.

\subsection{Two-sample testing using $f_\lambda$}

We now discuss in more detail about testing using the $f_\lambda$ statistic for the case of $\IER$ graphs. We state a result below that is along the lines of Corollary~\ref{cor_fDelta_IER}, but the main objective of this part is to demonstrate there are cases, where one needs to expand the null hypothesis by allowing $\epsilon(Q,Q')$ to be a positive function.

The situation typically arises if we deal with graphs of different sizes. For instance, let $Q,Q'$ correspond to ER models with graph sizes $n,n'$ and edge probabilities $p,p'$, respectively.
Further, assume that we consider only the largest singular value, that is, $f_\lambda = \frac{\lambda_1(A_G)}{n}$. As a consequence,
$\mu(Q) = p - \frac{p}{n}$, and $\mu(Q') = p' - \frac{p'}{n'}$. 
Typically, we would like to call the distributions $Q,Q'$ same if $p=p'$, irrespective of the the graph sizes, but in this case, if $n\neq n'$, then one can see that $\mu(Q)\neq\mu(Q')$ even if $p=p'$.
This observation suggests that for the case of ER graphs, we should not distinguish between distributions for which $\epsilon(Q,Q') = O\left(\min\left\{\frac1n,\frac{1}{n'}\right\}\right)$.

One can easily imagine other scenarios, in particular for stochastic block models, where similar situations arise. The following result presents a general guarantee in this setting. 

\begin{corollary}[Consistent test for semi-sparse $\IER$ using $f_\lambda$]
\label{cor_flambda_IER}
Consider the setting given in Lemma~\ref{lem_flambda_assum}. Then $f_\lambda$ satisfies Assumption~\ref{assum_est} for $\widehat\sigma(G) = \frac{2.1}{n}\sqrt{k\widehat{D}(G)}$, where $n$ is the number of vertices in $G$ and $\widehat{D}(G)$ is its maximum degree, and $k$ is the number of largest singular values computed in $f_\lambda$. 
Hence, the proposed test is uniformly consistent for the testing problem with any
\begin{equation*}
 \epsilon(Q,Q') \leq \frac{C}{\min\{n,n'\}}
 \text{ and }
 \rho(Q,Q') \geq {7.5\sqrt{k}}\left( \frac{\sqrt{D_Q}}{n} + \frac{\sqrt{D_{Q'}}}{n'}\right)
\end{equation*} 
where $C>0$ is any absolute constant, and $Q\in\Fc_f\cap\Fc_n$ and $Q'\in\Fc_f\cap\Fc_{n'}$.
\end{corollary}

We note that based on Theorem~\ref{thm_main}, the allowable upper limit of $\epsilon$ could be increased, but may not be needed from a practical perspective.

\section{Minimax optimality}
\label{sec_minimax}

The purpose of this section is to prove that the separation conditions on $\rho$ stated in Corollaries~\ref{cor_fDelta_IER} and~\ref{cor_flambda_IER} are necessary for testing between the two alternatives.
This implies that the proposed two-sample test is optimal when we restrict to $\IER$ graphs, and consider network statistics $f_\Delta$ or $f_\lambda$.
Instead of directly stating the converse of  Corollaries~\ref{cor_fDelta_IER} and~\ref{cor_flambda_IER}, we digress a little to study the total variation distance between a particular pair of network models.

Let  $(p_n)_{n\geq 1}\subset[0,1]$ be a sequence of probabilities, and $(\gamma_n)_{n\geq 1}\subset(0,1)$ be a sequence of small positive values such that $\gamma_n\leq \min\{p_n,1-p_n\}$.
We will consider two sequences of models on $\Fc_{\geq1}$. The first one consists of ER models $(Q_n)_{n\geq1}$, where $Q_n\in\Fc_{2n}$ with edge probability $p_n$.
Note that we consider only graphs on even numbers of vertices.
The second sequence, $(Q'_n)_{n\geq1}$ consists of mixture distributions defined as follows. For each $n$, let $\ell\in\{-1,+1\}^{2n}$ be a balanced labelling of $2n$ vertices, that is, $\sum_i \ell_i = 0$.
For any such labelling, there is a model such that edges between vertices of same label occur with probability $(p_n+\gamma_n)$, while edges between vertices with different labels occur with probability $(p_n-\gamma_n)$.
There are $\binom{2n}{n}$ possible balanced labellings, and $Q'_n$ is an uniform mixture of the models obtained with different labellings.
In other words, we generate a graph from $Q'_n$ as follows. First we randomly, and uniformly, choose any balanced labelling of $2n$ vertices, and then generate a random graph with two chosen communities, and aforementioned edge probabilities.
The following result states that this pair of sequences are asymptotically indistinguishable for small $\gamma_n$.
\begin{theorem}[Distance between $Q_n$ and $Q'_n$]
\label{thm_separation}
If $\gamma_n = o_n\left(\sqrt{\frac{p_n(1-p_n)}{n}}\right)$, then the total variation distance
\begin{equation*}
 \Vert Q_n - Q'_n \Vert_{TV} = \frac12\sum_G \left| Q_n(G) - Q'_n(G) \right|  \in o_n(1) \;,
\end{equation*}
 where the sum is taken over all possible graphs on $2n$ vertices.
\end{theorem}

The above result can be of independent interest, particularly when one observes that $Q'_n$ is the uniform distribution on all stochastic block models with two balanced communities, and edge probabilities $(p_n+\gamma_n)$ and $(p_n - \gamma_n)$. 
For this problem, it is well known that the partitions can be identified if and only if $\gamma_n > C\sqrt{\frac{p_n(1-p_n)\ln n}{n}}$ for some constant $C>0$.~\citep[see, for instance,][]{Chen_2016_jour_JMLR}. 
Theorem~\ref{thm_separation} shows that without the $\ln n$ factor, one cannot even distinguish the planted graph from a purely random graph.
The above detection boundary also matches the fundamental limit of partial recovery of communities in sparse stochastic block models~\citep{Mossel_2015_jour_ProbTheory}.
Some results related to Theorem~\ref{thm_separation} can also be found in~\citet{AriasCastro_2014_jour_AnnStat} and~\citet{Chatterjee_2012_jour_AnnStat} for the problems of testing or estimation of $\IER$ models, and in~\citet{Carpentier_2015_jour_EJStat} and~\citet{Banks_2016_arxiv} for the case of signal detection. 

Our interest in Theorem~\ref{thm_separation} stems from the fact that it provides a ``hard'' instance in the context of testing with $f_\Delta$ or $f_\lambda$.  
To put this in perspective, let $p_n\leq \frac12$, and consider the two-sample testing problem, where the two graphs are generated from either of the above models.
Theorem~\ref{thm_separation} implies that if $\gamma_n = o_n\left(\sqrt{\frac{p_n}{n}}\right)$, then no two-sample test can achieve a low Type-I + Type-II error rate. 
On the other hand, a simple calculation combined with Corollary~\ref{cor_flambda_IER} shows that the proposed test with $f_\lambda$ statistic is consistent when $\gamma_n \geq 15\sqrt{\frac{p_n}{n}}$ (see Claim~\ref{claim_lambda2} for details). 
Similarly, one can also verify that the same test in combination with $f_\Delta$ is consistent for $\gamma_n \geq 5\sqrt{\frac{p_n\ln n}{n}}$ (see Claim~\ref{claim_Delta2}). 
Thus, our general testing principle provides a test based on $f_\lambda$ that can detect separation near the fundamental limit of distinguishability, whereas its combination with $f_\Delta$ is only worse by a logarithmic factor.  

The above discussion, when stated formally, provides the following results that, in conjunction with Corollaries~\ref{cor_fDelta_IER} and~\ref{cor_flambda_IER}, guarantee the minimax (near) optimality of
the tests based on $f_\Delta$ and $f_\lambda$.

\begin{corollary}[Minimax separation for testing semi-sparse $\IER$ using $f_\Delta$]
\label{cor_fDelta_minimax}
Consider the setting of Lemma~\ref{lem_fDelta_IER_assum} and Corollary~\ref{cor_fDelta_IER}, and without loss of generality, let $n\leq n'$. If
\begin{equation*}
 \rho(Q,Q') = o_n\left( \sqrt{\frac{\mu(Q)}{\binom{n}{3}}} + \sqrt{\frac{\mu(Q')}{\binom{n'}{3}}} \right)
\end{equation*} 
for all $Q\in\Fc_f\cap\Fc_n$ and $Q'\in\Fc_f\cap\Fc_{n'}$, then for any test $\Psi$ and any $n_0$
\begin{align*}
\sup_{Q\in\Fc_{f} \cap \Fc_{\geq n_0}} \bigg( 
\sup_{Q'\in\Bc(Q,\epsilon)\cap \Fc_{\geq n_0}} &\P_{G\sim Q,G'\sim Q'}(\Psi \text{ rejects } H_0) ~+ 
\\& \qquad
\sup_{Q'\in\overline{\Bc(Q,\rho)}\cap \Fc_{\geq n_0}}\P_{G\sim Q,G'\sim Q'}(\Psi \text{ accepts } H_0) \bigg)
= 1 \;.
\end{align*}
\end{corollary}

To put it simply, if the separation condition is Corollary~\ref{cor_fDelta_IER} is not satisfied (ignoring logarithmic difference), then no matter how large the graphs are, one cannot hope to achieve a bounded error rate with any two-sample test.
The corresponding optimality result for $f_\lambda$ with $k=2$ is stated below.

\begin{corollary}[Minimax separation for testing semi-sparse $\IER$ using $f_\lambda$]
\label{cor_flambda_minimax}
Consider the setting of Lemma~\ref{lem_flambda_assum} and Corollary~\ref{cor_flambda_IER}, where we consider only the largest two singular values, that is, $k=2$. Let
\begin{equation*}
 \rho(Q,Q') = o_n\left( \frac{\sqrt{D_Q}}{n} + \frac{\sqrt{D_{Q'}}}{n'} \right)
\end{equation*} 
for all $Q\in\Fc_f\cap\Fc_n$ and $Q'\in\Fc_f\cap\Fc_{n'}$, and $\epsilon(Q,Q')$ be bounded as in Corollary~\ref{cor_flambda_IER}. Then for any test $\Psi$ and any $n_0$
\begin{align*}
\sup_{Q\in\Fc_{f} \cap \Fc_{\geq n_0}} \bigg( 
\sup_{Q'\in\Bc(Q,\epsilon)\cap \Fc_{\geq n_0}} &\P_{G\sim Q,G'\sim Q'}(\Psi \text{ rejects } H_0) ~+ 
\\& \qquad
\sup_{Q'\in\overline{\Bc(Q,\rho)}\cap \Fc_{\geq n_0}}\P_{G\sim Q,G'\sim Q'}(\Psi \text{ accepts } H_0) \bigg)
= 1 \;.
\end{align*}
\end{corollary}

Observe that the above result is sharp since it precisely matches the sufficient condition of Corollary~\ref{cor_flambda_IER}. 
We feel that a similar result can be shown for any $k>2$ by using a generalisation of Theorem~\ref{thm_separation} for $k$ community models.

\section{Proofs}
\label{sec_proof}

Here, we sequentially present proofs of all the results stated in the previous sections.

\subsection*{Proof of Lemma~\ref{lem_fDelta_assum}}
Observe that $\binom{n}{3}f_\Delta(G) = \sum\limits_{i<j<k} \Delta_{ijk}$, and 
\begin{align*}
\mathsf{Var}_{G\sim Q}\left(\binom{n}{3}f_\Delta(G)\right)
&=  \sum\limits_{i'<j'<k'}\sum\limits_{i<j<k}  \left(\E_{G\sim Q}[\Delta_{ijk} \Delta_{i'j'k'}] - \E_{G\sim Q}[\Delta_{ijk}]\E_{G\sim Q}[ \Delta_{i'j'k'}]\right)\;.
\end{align*}
Consider the inner sum for $(i',j',k') = (1,2,3)$.
Under the condition that non-overlapping triangles are uncorrelated, one can see that the terms where at most one of $i,j,k$ is in $\{1,2,3\}$ are zero. So
\begin{align}
&\sum\limits_{i<j<k}  \left(\E_{G\sim Q}[\Delta_{ijk} \Delta_{123}] - \E_{G\sim Q}[\Delta_{ijk}]\E_{G\sim Q}[ \Delta_{123}]\right)
\nonumber
\\&=\left(\E_{G\sim Q}[\Delta_{123}]- (\E_{G\sim Q}[\Delta_{123}])^2\right) +
\sum_{\substack{1\leq i<j\leq 3\\ k>3}}\left(\E_{G\sim Q}[\Delta_{ijk} \Delta_{123}] - \E_{G\sim Q}[\Delta_{ijk}]\E_{G\sim Q}[ \Delta_{123}]\right)
\nonumber
\\&\leq \E_{G\sim Q}[\Delta_{123}] + \sum_{\substack{1\leq i<j\leq 3\\ k>3}}\E_{G\sim Q}[\Delta_{ijk} \Delta_{123}]
\nonumber
\\&= \E_{G\sim Q}[\Delta_{123}] + \sum_{k>3}\E_{G\sim Q}\left[\Delta_{123}\left((A_G)_{1k}(A_G)_{2k}+(A_G)_{2k}(A_G)_{3k}+(A_G)_{1k}(A_G)_{3k}\right)\right]
\nonumber
\\&\leq \E_{G\sim Q}[\Delta_{123}] + \sum_{k>3}\E_{G\sim Q}\left[\Delta_{123}\left((A_G)_{1k}+(A_G)_{2k}+(A_G)_{3k}\right)\right]\;.
\label{eqn_lem1_pf1}
\end{align}
We now use the fact that $\Delta_{123}$ is not correlated with any non-overlapping edge to decompose the expectation. Subsequently, summing over $k$ shows that the above quantity can be bounded from above by $\E_{G\sim Q}[\Delta_{123}](1+3D_Q)$, where $D_Q = \max_i \sum_k \E[(A_G)_{ik}]$ is the maximum expected degree. Summing over all $i',j',k'$ leads to the conclusion
\begin{align*}
\mathsf{Var}_{G\sim Q}\left(\binom{n}{3}f_\Delta(G)\right)
&\leq \binom{n}{3}\mu(Q)(1+3D_Q)\;,
\end{align*}
and by Chebyshev's inequality,
\begin{equation*}
\P_{G\sim Q}( | f_\Delta(G) - \mu(Q) | > \sigma(Q)) \leq \frac{(3D_Q +1)\binom{n}{3}\mu(Q)}{{\binom{n}{3}}^2\sigma(Q)^2} = \frac{1}{\ln n} = o_n(1)\;.
\end{equation*}

\subsection*{Proof of Lemma~\ref{lem_fDelta_Geom_assum}}
This result follows by observing that in a random geometric graph two non-overlapping triangles or edges are uncorrelated. Also, $D_Q = (n-1)p_Q <np_Q$.

\subsection*{Proof of Lemma~\ref{lem_fDelta_IER_assum}}
We may proceed similar to the proof of Lemma~\ref{lem_fDelta_assum} up to the last but one step of~\eqref{eqn_lem1_pf1}, which can be written as
\begin{align*}
\E_{G\sim Q}[\Delta_{123}] + \E_{G\sim Q}[\Delta_{123}]\sum_{k>3}\left((M_Q)_{1k}(M_Q)_{2k}+(M_Q)_{2k}(M_Q)_{3k}+(M_Q)_{1k}(M_Q)_{3k}\right)\;,
\end{align*}
and is smaller that $4\E_{G\sim Q}[\Delta_{123}]$ since $Q\in\Fc_f$ defined in Lemma~\ref{lem_fDelta_IER_assum}. Thus, 
$\mathsf{Var}_{G\sim Q}\left(\binom{n}{3}f_\Delta(G)\right)
\leq 4\mu(Q)\binom{n}{3}$, and the result follows due Chebyshev's inequality and the stated choice of $\sigma$.

\subsection*{Proof of Lemma~\ref{lem_flambda_assum}}
Let $M_Q$ and $D_Q$ in the statement of the lemma, and let $d$ be the Euclidean metric in $\mathbb{R}^k$. Then, from Weyl's inequality, we have
\begin{align*}
 d(f_\lambda(G),\mu(Q))  = \frac1n \sqrt{\sum_{i\leq k} (\lambda_i(A_G) - \lambda_i(M_Q))^2} 
 \leq \frac{\sqrt{k}}{n} \Vert A_G - M_Q\Vert \;,
\end{align*}
where $\Vert\cdot\Vert$ is the spectral norm.
\citet{Lu_2013_jour_EJComb} show that if $D_Q \geq C(\ln n)^4$ holds for a large constant $C>0$, then $\Vert A_G-M_Q\Vert \leq (2+\delta)\sqrt{D_Q}$ with probability $1-o_n(1)$ for any $\delta>0$.
This immediately leads to the choices of $\Fc_f$ and $\sigma(Q)$ stated in the lemma.

\subsection*{Proof of Theorem~\ref{thm_main}}
Let $Q\in \Fc_f$ and $Q'\in\Bc(Q,\epsilon)$, where $\epsilon(Q,Q')\leq 0.5(\sigma(Q)+\sigma(Q'))$. Observe that
\begin{align*}
&\P_{G\sim Q, G'\sim Q'} (T(G,G')>1)
\\&= \P_{G\sim Q, G'\sim Q'} \big( d(f(G), f(G')) >2(\widehat\sigma(G) + \widehat\sigma(G'))\big)
\\ &\leq
\P_{G\sim Q, G'\sim Q'} \big( d(f(G), f(G')) >1.5(\sigma(G) + \sigma(G'))\big) 
\\ & \qquad
+\P_{G\sim Q, G'\sim Q'} \big( 1.5(\sigma(G) + \sigma(G')) >2(\widehat\sigma(G) + \widehat\sigma(G'))\big)
\\ &\leq
\P_{G\sim Q, G'\sim Q'} \big( d(f(G), f(G')) >1.5(\sigma(G) + \sigma(G'))\big) 
\\ & \qquad
+\P_{G\sim Q} \left( \sigma(G) - \widehat\sigma(G) >\textstyle\frac14\sigma(G)\right)
+\P_{G'\sim Q'} \left( \sigma(G') - \widehat\sigma(G') >\textstyle\frac14\sigma(G')\right)
\;.
\end{align*}
Due to Assumption~\ref{assum_est}, the second and third terms in the bound are both $o_n(1)$.
To bound the the first term, we note that
\begin{align*}
d(f(G), f(G')) &\leq d(f(G), \mu(Q)) + d(f(G'), \mu(Q')) + d(\mu(Q), \mu(Q'))
\\ &\leq d(f(G), \mu(Q)) + d(f(G'), \mu(Q')) + 0.5(\sigma(Q)+\sigma(Q'))\;.
\end{align*}
Using this fact, we have
\begin{align*}
\P_{G\sim Q, G'\sim Q'} &\big( d(f(G), f(G')) >1.5(\sigma(G) + \sigma(G'))\big) 
\\ &\leq 
\P_{G\sim Q, G'\sim Q'} \big( d(f(G), \mu(Q)) + d(f(G'), \mu(Q')) > \sigma(G) + \sigma(G')\big) 
\\ & \leq
\P_{G\sim Q} \big( d(f(G), \mu(Q)) > \sigma(G)\big) + \P_{G'\sim Q'} \big( d(f(G'), \mu(Q')) > \sigma(G')\big) \;,
\end{align*}
where both terms in the bounds are $o_n(1)$ due to Assumption~\ref{assum}. 
Thus, we have established that $\P_{G\sim Q, G'\sim Q'} (T(G,G')>1) \in o_n(1)$.

We tackle the other case, that is $Q'\in \overline{\Bc(Q,\rho)}$, in a similar way. Here, we have
\begin{align*}
\P_{G\sim Q, G'\sim Q'} &(T(G,G')\leq1)
\\ &\leq
\P_{G\sim Q, G'\sim Q'} \big( d(f(G), f(G')) \leq 2.5(\sigma(G) + \sigma(G'))\big) 
\\ & \qquad
+\P_{G\sim Q, G'\sim Q'} \big( 2.5(\sigma(G) + \sigma(G')) \leq 2(\widehat\sigma(G) + \widehat\sigma(G'))\big) \;,
\end{align*}
where the second term can be shown to be $o_n(1)$ using Assumption~\ref{assum_est}.
For the first term, we recall $\rho(Q,Q')\geq 3.5(\sigma(Q)+\sigma(Q'))$, and use triangle inequality in the opposite direction, that is,
\begin{align*}
3.5(\sigma(Q)+\sigma(Q')) 
&\leq d(\mu(Q), \mu(Q')) 
\\&\leq d(f(G), \mu(Q)) + d(f(G'), \mu(Q')) + d(f(G), f(G')) 
\end{align*}
Hence, the first term can be bounded as
\begin{align*}
\P_{G\sim Q, G'\sim Q'} &\big( d(f(G), f(G')) \leq 2.5(\sigma(G) + \sigma(G'))\big) 
\\ &\leq 
\P_{G\sim Q, G'\sim Q'} \big( d(f(G), \mu(Q)) + d(f(G'), \mu(Q')) > \sigma(G) + \sigma(G')\big) \;,
\end{align*}
which is again $o_n(1)$, thereby leading to a similar conclusion for Type-II error.

\subsection*{Proof of Corollary~\ref{cor_fDelta_IER}}
Note that we only need to show that $f_\Delta$ satisfies Assumption~\ref{assum_est} for $\widehat\sigma(G) = 2\sqrt{f_\Delta(G)\ln n/\binom{n}{3}}$.
The rest of the claim is an immediate consequence of Theorem~\ref{thm_main}.
To prove the first part, we deal with the upper and lower tail separately. For the upper tail, we observe that for any $\delta>0$,
\begin{align*}
\P_{G\sim Q}(\widehat\sigma(G) > (1+\delta)\sigma(Q))
&\leq \P_{G\sim Q}\left(f_\Delta(G) > (1+\delta)^2\mu(Q)\right)
\\&\leq \P_{G\sim Q}\left(f_\Delta(G)-\mu(Q) > \delta\mu(Q)\right)
\\&\leq \frac{\mathsf{Var}_{G\sim Q}(f_\Delta(G))}{\delta^2\mu^2(Q)} \in o_n(1)\;.
\end{align*}
The last step uses the bound on ${\mathsf{Var}_{G\sim Q}(f_\Delta(G))}$ derived in proof of Lemma~\ref{lem_fDelta_IER_assum} and the fact $\mu(Q)\geq \frac{\ln n}{n^3}$.
We can similarly bound the lower tail probability. 

\subsection*{Proof of Corollary~\ref{cor_fDelta_Geom}}
As in the previous proof, we only need to show that $f_\Delta$ satisfies Assumption~\ref{assum_est} in this setting, where $\widehat\sigma(G) = \sqrt{(3n\widehat{p}(G)+1)f_\Delta(G)\ln n/\binom{n}{3}}$.
We bound the upper tail probability as
\begin{align*}
\P_{G\sim Q}&(\widehat\sigma(G) > (1+\delta)\sigma(Q))
\leq \P_{G\sim Q}\left((3n\widehat{p}(G)+1)f_\Delta(G) > (1+\delta)^2(3np_Q+1)\mu(Q)\right)
\\&\leq \P_{G\sim Q}\left((3n\widehat{p}(G)+1) > (1+\delta)(3np_Q+1)\right) 
+  \P_{G\sim Q}\left(f_\Delta(G) > (1+\delta)\mu(Q)\right)
\end{align*}
using union bound. Observe that the first term is at most
\begin{align*}
\P_{G\sim Q}\left(\widehat{p}(G) > (1+\delta)p_Q\right) 
\leq \frac{\mathsf{Var}_{G\sim Q}(\widehat{p}(G))}{\delta^2p_Q^2} \leq \frac{1}{\delta^2 p_Q\binom{n}{2}} = o_n(1)\;.
\end{align*}
Here, we use Chebyshev's inequality followed by the fact that  $\mathsf{Var}_{G\sim Q}(\widehat{p}(G))={\binom{n}{2}}^{-1}p(1-p)$ since the edges of the geometric graph are still pairwise independent. Finally, we use the condition $p_Q \geq \frac{1}{n}$ for all $Q\in\Fc_f$ stated in Lemma~\ref{lem_fDelta_Geom_assum}.

For the second term, we use the variance bound stated in proof of Lemma~\ref{lem_fDelta_Geom_assum} to write
\begin{align}
\P_{G\sim Q}(f_\Delta(G) > (1+\delta)\mu(Q))
\leq \frac{\mathsf{Var}_{G\sim Q}(f_\Delta(G))}{(\delta\mu(Q))^2}
\leq \frac{3np_Q+1}{\delta^2\binom{n}{3}\mu(Q)} \;.
\label{eqn_Geom_cor1}
\end{align}
\citet[Lemma 1]{Bubeck_2016_jour_RSA} showed that if $p_Q\geq\frac14$ and $r_Q\geq C_p$ for some $C_p>0$ depending on $p_Q$, then $\mu(Q) \geq p_Q^3$. Taking $C=\sup\limits_{1/4\leq p\leq1} C_p$, we can say that if $p_Q\geq \frac14$ and $r_Q\geq C$, then $\mu(Q) \geq p_Q^3$. Note that this gives rise to the $C$ mentioned in Lemma~\ref{lem_fDelta_Geom_assum}. Hence, under this regime, the bound in~\eqref{eqn_Geom_cor1} is $o_n(1)$.

For $p_Q<\frac14$, which includes the case of $p_Q$ decaying with $n$, \citet{Bubeck_2016_jour_RSA} showed that $\mu(Q) \geq cp_Q^3(\ln \frac{1}{p_Q})^{3/2}r_Q^{-1/2}$ for some absolute constant $c>0$. In this case, one can verify that if $np_Q\geq1$ and the upper bound on $r_Q$ (see Lemma~\ref{lem_fDelta_Geom_assum}) hold, then the bound in~\eqref{eqn_Geom_cor1} is also $o_n(1)$. 
This takes care of the upper tail probability. For the other part, we have
\begin{align*}
\P_{G\sim Q}&(\widehat\sigma(G) < (1-\delta)\sigma(Q))
\leq \P_{G\sim Q}\left(\widehat{p}(G) < (1-\delta)p_Q\right) 
+  \P_{G\sim Q}\left(f_\Delta(G) < (1-\delta)\mu(Q)\right)\;,
\end{align*}
where the terms can be again bounded as above to show that the probability is $o_n(1)$.

\subsection*{Proof of Corollary~\ref{cor_fDelta_ERvGeom}}
Let $\sigma_E$ and $\sigma_G$ be the deviations defined in Lemmas~\ref{lem_fDelta_IER_assum} and~\ref{lem_fDelta_Geom_assum}, respectively. A natural choice for $\sigma$ in combined setting is $\sigma = \max\{\sigma_E,\sigma_G\}$, which can also be estimated accurately from the random graph. In other words, $f_\Delta$ satisfies Assumption~\ref{assum_est} in the present setting as well.
Also note that $p_n$ decays at the appropriate rate so that the ER graph lies is $\Fc_f$ defined in Lemma~\ref{lem_fDelta_IER_assum}. Also the decaying $p_n$ corresponds to the sparser regime discussed in proof of Corollary~\ref{cor_fDelta_Geom}.

Let $Q,Q'$ be the distributions of the two random graphs. If $Q=Q'$, then Theorem~\ref{thm_main} directly provides a bound on Type-I error rate as $\epsilon(Q,Q')=0$. Thus, we only need to check the case, where $Q\neq Q'$, or more precisely, we need to verify that if $r_n = o_n\left( (\ln \frac{1}{p_n})^3 \right)$, then we get $\rho(Q,Q') > 3.5(\sigma(Q)+\sigma(Q'))$.
Assume that $Q$ corresponds to ER graph, and $Q$ is the $\Geom$ graph, that is, $\mu(Q') = p_n^3$ and $\mu(Q) \geq cp_n^3(\ln \frac{1}{p_n})^{3/2}r_n^{-1/2}$.
It is now easy to verify that under prescribed condition on $r_n$, all $\mu(Q'),\sigma(Q')$ and $\sigma(Q)$ are $o_n(\mu(Q))$. Hence, the separation condition is eventually satisfied.

\subsection*{Proof of Corollary~\ref{cor_flambda_IER}}

Following the proof of Corollary~\ref{cor_fDelta_IER}, note that we only need to prove concentration of $\widehat\sigma$. The rest follows from Theorem~\ref{thm_main}. 
The only difference lies in the stated upper bound on $\epsilon$, which follows from the conditions $D_Q \geq (\ln n)^{4.1}$ and $D_{Q'} \geq (\ln n')^{4.1}$.

We now show that $\widehat\sigma(G) = \frac{2.1}{n}\sqrt{k\widehat{D}(G)}$ concentrates about $\sigma(Q)$.
Recall that $Q\in\IER\cap\Fc_n$ is characterised by the matrix $M_Q\in[0,1]^{n\times n}$, and $D_Q$ is the maximum row sum of $M_Q$. 
Let us denote the row sums of $M_Q$ by $D_1, D_2,\ldots, D_n$. 
Also for $G\sim Q$, denote the degree of vertex-$i$ by $\dg_i$. Hence, we have $\widehat{D}(G) = \max_i \dg_i$, and $\E_{G\sim Q}[\dg_i] = D_i$.
For $\delta\in(0,1)$, we bound the upper tail probability by
\begin{align}
\P_{G\sim Q}\left(\widehat\sigma(G) > (1+\delta)\sigma(Q)\right)
& \leq \P_{G\sim Q}\left(\widehat{D}(G) > (1+\delta)D_Q\right)
\nonumber
\\& = \P_{G\sim Q}\left(\bigcup_i \{ \dg_i > (1+\delta)D_Q\}\right)
\nonumber
\\& \leq \sum_i\P_{G\sim Q}\left( \dg_i > (1+\delta)D_Q\right)
\label{eqn_flambda_cor1}
\end{align}
using the union bound. 
Consider the cases $D_i \geq (\ln n)^{1.1}$ and $D_i < (\ln n)^{1.1}$ separately.
In the former case,
\begin{align*}
\P_{G\sim Q}\left( \dg_i > (1+\delta)D_Q\right)
& \leq \P_{G\sim Q}\left( \dg_i - D_i> \delta D_i\right)
\leq e^{-\delta^2D_i/3} \;,
\end{align*}
where we use the Bernstein inequality at the last step. Since, $D_i \geq (\ln n)^{1.1} \geq \frac{6}{\delta^2}\ln n$ for large $n$, the above probability is bounded by $\frac{1}{n^2}$.
On the other hand, if $D_i\leq (\ln n)^{1.1}$, we use the Markov inequality to write
\begin{align*}
\P_{G\sim Q}\left( \dg_i > (1+\delta)D_Q\right)
& \leq \P_{G\sim Q}\left( e^{\dg_i}> e^{D_Q}\right)
\\&\leq e^{-D_Q}\prod_j \E_{G\sim Q}\left[e^{(A_G)_{ij}} \right];,
\\&\leq e^{-D_Q} \prod_j \left( 1 + e(M_Q)_{ij} \right)
\leq e^{eD_i - D_Q},
\end{align*}
where we use independence of the edges for second inequality, and the fact $(1+x)\leq e^x$ in the last step. 
Since $D_i<(\ln n)^{1.1}$ and $D_Q\geq (\ln n)^{4.1}$, the above bound is eventually smaller than $\frac{1}{n^2}$.
From above arguments, we can see that each term in~\eqref{eqn_flambda_cor1} is at most $\frac{1}{n^2}$, and hence the sum is $o_n(1)$.

To prove the lower tail bound, assume without loss of generality that $D_Q=D_1$. We can see
\begin{align*}
\P_{G\sim Q}\left(\widehat\sigma(G) < (1-\delta)\sigma(Q)\right)
& \leq \P_{G\sim Q}\left(\widehat{D}(G) < (1-\delta)D_Q\right)
\\& = \P_{G\sim Q}\left(\bigcap_i \{ \dg_i < (1-\delta)D_Q\}\right)
\\& \leq \P_{G\sim Q}\left( \dg_1 < (1-\delta)D_Q\right)
\\&= \P_{G\sim Q}\left( D_1 - \dg_1 > \delta D_1\right)
\leq e^{-\delta^2D_1/3} \;,
\end{align*}
where the last bound is due to Bernstein inequality, and is $o_n(1)$ under the condition on $D_Q$.
Hence, $f_\lambda$ satisfies Assumption~\ref{assum_est}, and the claim follows.

\subsection*{Proof of Theorem~\ref{thm_separation}}

For ease of notation, we drop the subscript $n$ from $Q_n, Q'_n, p_n$ and $\gamma_n$, but we recall that the subscript $n$ corresponds to graphs of size $2n$, and in the case of $Q'_n$, there are exactly $n$ vertices labelled $+1$, and the rest $-1$.
We first define the quantity
\begin{displaymath}
L(Q,Q') = \sum_G \frac{(Q'(G))^2}{Q(G)} \;,
\end{displaymath}
and use Cauchy-Schwarz inequality to bound the total variation distance by
\begin{align*}
\Vert Q- Q'\Vert_{TV}
\leq \frac12 \sqrt{\sum_G Q(G)} \sqrt{\sum_G Q(G) \left(\frac{Q'(G)}{Q(G)} - 1\right)^2}
= \frac{1}{2} \sqrt{L(Q,Q')-1} \;,
\end{align*}
where we use the fact $\sum_G Q(G) = \sum_G Q'(G) = 1$.
If we can show $L(Q,Q') \leq e^{o_n(1)}$, then the abound bound implies 
$\Vert Q- Q'\Vert_{TV} \in o_n(1)$, which is the claim.
Hence, the rest of the proof is about deriving the stated upper bound for $L(Q,Q')$. For this, observe that given a balanced labelling $\ell\in\{-1,+1\}^{2n}$, the conditional distribution of $Q'$ is given by
\begin{align*}
Q'(G | \ell) = \prod_{i<j} (p+\gamma \ell_i\ell_j)^{(A_G)_{ij}}(1-p-\gamma \ell_i\ell_j)^{1-(A_G)_{ij}} \;,
\end{align*}
where $A_G$ is the adjacency matrix of $G$. As a consequence, the distribution $Q'$ is given by
\begin{align*}
Q'(G) = \frac{1}{\binom{2n}{n}} \sum_\ell \prod_{i<j} (p+\gamma \ell_i\ell_j)^{(A_G)_{ij}}(1-p-\gamma \ell_i\ell_j)^{(A_G)_{ij}} \;,
\end{align*}
where the sum is over all balanced labellings. Hence, one can compute $L(Q,Q')$ as
\begin{align*}
L&(Q,Q')= 
\\&
\sum_G \frac{1}{{\binom{2n}{n}^2}}\sum_{\ell,\ell'} \prod_{i<j} \left( \frac{(p+\gamma \ell_i\ell_j)(p+\gamma \ell'_i\ell'_j)}{p}\right)^{(A_G)_{ij}}
\left( \frac{(1-p-\gamma \ell_i\ell_j)(1-p-\gamma \ell'_i\ell'_j)}{1-p}\right)^{1-(A_G)_{ij}}
\end{align*}
Interchanging the sums and summing the products over all possible graphs $G$, we obtain
\begin{align*}
L(Q,Q') = \frac{1}{{\binom{2n}{n}^2}}\sum_{\ell,\ell'} \prod_{i<j}\left(1 + \frac{\gamma^2}{p(1-p)}\ell_i\ell_j\ell'_i\ell'_j\right)\;.
\end{align*}
Note that due to the symmetric nature of $\ell$ and $\ell'$, if we fix an $\ell'$ and sum over $\ell$, then the sum remains same irrespective of the value of $\ell'$.
Hence, we may fix $\ell'$ to the vector with the first $n$ coordinates as $+1$, and the rest $-1$, and consider the average only over $\ell$, that is,
\begin{align*}
L(Q,Q') 
&= \frac{1}{\binom{2n}{n}}\sum_{\ell} \prod_{\substack{ i<j\leq n\\ n<i< j}} \left(1 + \frac{\gamma^2}{p(1-p)}\ell_i\ell_j\right)  \prod_{i\leq n< j} \left(1 - \frac{\gamma^2}{p(1-p)}\ell_i\ell_j\right)
\\&\leq \frac{1}{\binom{2n}{n}}\sum_{\ell} \exp\left( \sum_{\substack{ i<j\leq n\\ n<i< j}}  \frac{\gamma^2}{p(1-p)}\ell_i\ell_j -  \sum_{i\leq n< j} \frac{\gamma^2}{p(1-p)}\ell_i\ell_j\right)
\\&= \E_{\ell} \left[ \exp\left( \sum_{\substack{ i<j\leq n\\ n<i< j}}  \frac{\gamma^2}{p(1-p)}\ell_i\ell_j -  \sum_{i\leq n< j} \frac{\gamma^2}{p(1-p)}\ell_i\ell_j\right) \right] \;.
\end{align*}
The inequality follows from the relation $1+x\leq e^x$, and in the subsequent step, we view $\ell$ as a random labelling chosen uniformly from all balanced labellings.
We define the quantity $S_k = \sum\limits_{i=1}^k \ell_i$, and observe that $S_{2n} = 0$, which also implies $\sum_{j>n} \ell_j = -S_n$. Using this relation, the above bound simplifies to
\begin{align}
L(Q,Q')  \leq \exp\left( - \frac{\gamma^2n}{p(1-p)}\right)\E_{\ell_1,\ldots,\ell_n} \left[\exp\left( \frac{2\gamma^2}{p(1-p)} S_n^2\right)\right] 
\label{eqn_thm_sep_proof1}
\end{align}
We now observe that $\ell_1,\ldots,\ell_{2n}$ are conditional Bernoulli's with the constraint$S_{2n}=0$, or exactly $n$ of them can be one. Hence, they can be generated using the following procedure. 
\begin{itemize}
\item $\ell_1$ takes values $+1$ or $-1$ with equal probability, and
\item for $k\geq 1$, if $N_{k}$ is the number of $+1$'s observed in the first $k$ coordinates, then
$\ell_{k+1}$ takes value $+1$ with probability $\frac{n-N_{k}}{2n-k}$.
\end{itemize}  
Based on this observation, we make the following claim, that we prove at the end of this section.
\begin{claim}
\label{claim1}
For any $1\leq k<n$ and any $x>0$ such that $|2xS_k|\leq 1$, we have 
\begin{displaymath}
\E_{\ell_{k+1}|\ell_1,\ldots,\ell_k} \left[\exp\left(xS_{k+1}^2\right)\right] \leq \exp(x)\exp\left((x+2ex^2)S_k^2\right)
\end{displaymath}
\end{claim}
Now, define the sequence $(x_k)_{k=0,\ldots,n-1}$ such that $x_0 = \frac{2\gamma^2}{p(1-p)}$, and 
$x_{k+1} = x_k + 2ex_k^2$ for $k\geq0$.
Observe that under the condition on $\gamma$ in Theorem~\ref{thm_separation}, we eventually have that $x_0 \in [0, \frac{1}{8en}]$, and due to this, the sequence $(x_k)_k$ satisfies the following property.
\begin{claim}
\label{claim2}
If $x_0 \in [0, \frac{1}{8en}]$ and $x_{k+1} = (x_k + 2ex_k^2)$ for $k=0,\ldots,n-2$, then 
\begin{displaymath}
x_k \leq \left(1+ \frac{k}{n-1}\right) x_0 \leq 2x_0
\end{displaymath}
for every $k=0,\ldots,\leq n-1$. As a consequence, $\sum\limits_{k=0}^{n-1} x_k \leq 2nx_0$. 
\end{claim}
We return to~\eqref{eqn_thm_sep_proof1}, and use Claim~\ref{claim1} to bound $L(Q,Q')$ iteratively as
\begin{align*}
L(Q,Q') 
&\leq \exp( - nx_0/2)\E_{\ell_1,\ldots,\ell_{n-1}} \left[\E_{\ell_n|\ell_1,\ldots,\ell_{n-1}} \left[\exp\left( x_0 S_n^2\right)\right]\right] 
\\&\leq \exp( - nx_0/2)\exp(x_0)\E_{\ell_1,\ldots,\ell_{n-2}} \left[\E_{\ell_{n-1}|\ell_1,\ldots,\ell_{n-2}} \left[\exp\left( x_1 S_{n-1}^2\right)\right]\right] 
\\&\leq \exp( - nx_0/2)\exp\left(\sum_{i=0}^{k-1}x_i\right)\E_{\ell_1,\ldots,\ell_{n-k-1}} \left[\E_{\ell_{n-k}|\ell_1,\ldots,\ell_{n-k-1}} \left[\exp\left( x_k S_{n-k}^2\right)\right]\right] 
\\&\leq \exp( - nx_0/2)\exp\left(\sum_{i=0}^{n-2}x_i\right)\E_{\ell_1} \left[\exp\left( x_{n-1} S_1^2\right)\right]\;,
\end{align*}
where the condition in Claim~\ref{claim1} is satisfied at every step since $|2x_kS_{n-k-1}| \leq 4x_0|S_n| \leq 1$ as we have observed that $x_0 \leq \frac{1}{8en}$ eventually for large $n$.
Finally, note that $S_1^2 = \ell_1^2 = 1$, and hence, using Claim~\ref{claim2} we have the bound
\begin{align*}
L(Q,Q') \leq e^{2nx_0} = e^{4n\gamma^2/p(1-p)} = e^{o_n(1)}
\end{align*}
for $\gamma = o_n\left(\sqrt{\frac{p(1-p)}{n}}\right)$. Thus, we have the stated result.
We conclude this proof with the proof of the two intermediate claims.

\begin{proof}\textbf{of Claim~\ref{claim1}.~}
Recall the generation process for $\ell_k$, and observe the $S_k = N_k - (k-N_k) = 2N_k-k$.
Hence, for every $k\geq0$, the Bernoulli variable $\ell_{k+1}$ takes the value $+1$ with probability $\frac{n-N_{k}}{2n-k} = \frac12 - \frac{S_k}{2(2n-k)}$, and the value $-1$ with probability $\frac12 + \frac{S_k}{2(2n-k)}$. We evaluate the conditional expectation as follows
\begin{align*}
\E_{\ell_{k+1}|\ell_1,\ldots,\ell_k} &\left[\exp\left(xS_{k+1}^2\right)\right] 
= \E_{\ell_{k+1}|\ell_1,\ldots,\ell_k} \left[\exp\left(x(S_{k}^2 + 2S_k\ell_{k+1} + 1)\right)\right]
\\=& \exp\left(x+xS_k^2\right)\E_{\ell_{k+1}|\ell_1,\ldots,\ell_k}\left[\exp\left(2xS_k\ell_{k+1}\right)\right]
\\=& \exp\left(x+xS_k^2\right)\left(\exp(2xS_k)\left(\frac12 -  \frac{S_k}{2(2n-k)}\right) + \exp(-2xS_k)\left(\frac12 +  \frac{S_k}{2(2n-k)}\right)\right)
\\=& \exp\left(x+xS_k^2\right)\left(\cosh(2xS_k) - \frac{S_k}{(2n-k)}\sinh(2xS_k)\right) \;.
\end{align*}
One can verify that $z\sinh(z)$ is always positive, while $\cosh(z) \leq 1+ \frac{ez^2}{2} \leq \exp\left(\frac{ez^2}{2}\right)$ for all $|z|\leq1$.
As a consequence, the second term in above expression is positive, and can be ignored for an upper bound, whereas the first term is at most $\exp(2ex^2S_k^2)$.
Hence, the claim.
\end{proof} 

\begin{proof}\textbf{of Claim~\ref{claim2}.~}
We prove the claim by induction. Assume $x_k \leq \left( 1+ \frac{k}{n-1}\right)x_0$. Then
\begin{align*}
x_{k+1} &= x_k+2ex_k^2
\\&\leq x_0\left( 1+ \frac{k}{n-1}\right) + 2ex_0^2\left( 1+ \frac{k}{n-1}\right)^2
\\&\leq x_0\left( 1+ \frac{k}{n-1}\right) + \frac{1}{4(n-1)}\cdot4x_0
\\&= \left( 1+ \frac{k+1}{n-1}\right)x_0 \;,
\end{align*}
where the second inequality uses the facts $2ex_0 \leq \frac{1}{4(n-1)}$ and $(1+\frac{k}{n-1})\leq 2$ for $k<n$.
\end{proof}

\subsection*{Proof of Corollary~\ref{cor_fDelta_minimax}}
We begin by noting that both sequence of models in Theorem~\ref{thm_separation} belong to $\Fc_f$ if the sequences $p_n, \gamma_n$ satisfy $\frac{\ln 2n}{2n} \leq p_n-\gamma_n < p_n+\gamma_n \leq \frac{1}{\sqrt{2n}}$. 
As in Theorem~\ref{thm_separation}, let $Q_n$ and $Q'_n$ respectively denote the ER and the labelled graph models. 
Due to condition on $\rho$ given in statement of Corollary~\ref{cor_fDelta_minimax}, we can define a sequence $t_n \in o_n(1)$ such that
\begin{equation}
\rho(Q_n,Q'_n) = t_n\cdot \left( \sqrt{\frac{\mu(Q_n)}{n^3}} + \sqrt{\frac{\mu(Q'_n)}{n^3}} \right)\;,
\label{eqn_proof_fDelta_minimax1}
\end{equation}
where $Q_n,Q'_n$ are the above mentioned distributions.
We make the following claim.

\begin{claim}
\label{claim_Delta1}
Let $p_n\in[\frac{\ln 2n}{2n}, \frac{1}{\sqrt{2n}}]$. There exists $\tau_n \in o_n(1)$ that satisfies both the following conditions: (i) $\gamma_n = \tau_n\sqrt{\frac{p_n(1-p_n)}{n}}$, and (ii) $Q'_n \in \overline{\Bc(Q_n,\rho)}$ where $\rho$ is given by~\eqref{eqn_proof_fDelta_minimax1}.
\end{claim}
The above claim shows that, in the present scenario, we can define the sequences $Q_n,Q'_n$ such that $\Vert Q_n - Q'_n\Vert_{TV} \in o_n(1)$, and yet the pair satisfies the alternative hypothesis $H_1$.
This is the trick we use to prove that no non-trivial two-sample test that can distinguish between these two models.

At this stage, we follow the proof technique of~\citet{Collier_2012_conf_AISTAT} for lower bounding the error rate of a two-sample test by the error rate of a suitably defined one-sample test.
In our case, the one-sample test is: Given a graph $G$, identify whether $G$ is sampled from $Q_n$ or $Q'_n$.
Now, let $\Psi$ be any test for our two-sample testing problem. One can use $\Psi$ to construct a one-sample test for the above problem by comparing $G$ with a randomly generated graph from $Q_n$.
\citet{Collier_2012_conf_AISTAT} argues that this construction leads to a test whose Type-I and Type-II errors are both smaller than that of the of the two-sample test $\Psi$. Subsequently, using a standard testing lower bound~\citep[see][Section 7.1]{Baraud_2002_jour_Bernoulli}, one can lower bound the total error rate by $1-\Vert Q_n - Q'_n\Vert_{TV}$. 
Since, our $\Fc_f$ contain $Q_n,Q'_n$ for all large $n\geq n_0$, hence we can choose $n$ large enough to get the lower bound arbitrarily close to 1. Hence, we get the claimed supremum of 1 for any test $\Psi$.

While above arguments conclude the proof of Corollary~\ref{cor_fDelta_minimax}, for completeness, we also add the proof of the fact that for larger $\gamma_n$, the proposed two-sample test can indeed distinguish between $Q_n$ and $Q'_n$. In particular, let $\rho(Q_n,Q'_n) = 7\sqrt{\frac{\ln 2n}{\binom{2n}{3}}}(\sqrt{\mu(Q_n)}+\sqrt{\mu(Q'_n)})$. We prove the following.
\begin{claim} 
\label{claim_Delta2}
Let $p_n\in[\frac{\ln 2n}{2n}, \frac{1}{\sqrt{2n}}]$. If $\gamma_n \geq 5\sqrt{\frac{p_n\ln n}{n}}$, then $Q'_n \in \overline{\Bc(Q_n,\rho)}$ for above mentioned $\rho$. 
As a consequence, Corollary~\ref{cor_fDelta_IER} implies that proposed test consistently  distinguishes $Q_n$ from $Q'_n$.
\end{claim}
We now prove the claims.

\begin{proof}\textbf{of Claim~\ref{claim_Delta1}.~}
For convenience, we drop the subscript $n$.
One can easily verify that for the specified models and network statistic $f_\Delta$,
$\mu(Q) = p^3$ and $\mu(Q') = p^3 + \gamma^3 - \frac{3}{2n-1}(p^2\gamma + p\gamma^2)$.
Since $p\in o_n(1)$, we can abuse our notation to write $\gamma = \tau\sqrt{\frac{p}{n}}$ dropping the factor of $(1-p)$.
Now assume $\tau > \sqrt{12p}$, then we compute
\begin{align*}
|\mu(Q) - \mu(Q')| = \left| \frac{\tau^3p^{3/2}}{n^{3/2}} - \frac{3\tau p^{5/2}}{(2n-1)n^{1/2}} - \frac{3\tau^2 p^2}{(2n-1)n} \right|
> \frac{\tau^3p^{3/2}}{2n^{3/2}} \;,
\end{align*}
using the fact that $\tau > \sqrt{12p}$ ensures that the second and third terms are smaller that $\frac14$ of the first term. On the other hand, we can bound $\rho$ in~\eqref{eqn_proof_fDelta_minimax1} from above by
\begin{align*}
\rho(Q,Q') &= \frac{t}{n^{3/2}}\left(p^{3/2} + \sqrt{p^3 + \frac{\tau^3p^{3/2}}{n^{3/2}} - \frac{3\tau p^{5/2}}{(2n-1)n^{1/2}} - \frac{3\tau^2 p^2}{(2n-1)n}}\right) 
\\&\leq \frac{tp^{3/2}}{n^{3/2}} \left(1+ \sqrt{1+ \frac{\tau^3}{(pn)^{3/2}}}\right)
< \frac{3tp^{3/2}}{n^{3/2}}
\end{align*}
since $\tau\leq1$ and $np\geq1$. Hence, we can conclude that if $\tau>\max\{\sqrt{12p},\sqrt[3]{6t}\}$,
then $Q'\in\overline{\Bc(Q,\rho)}$. Since $p,t\in o_n(1)$, this is satisfied by some $\tau\in o_n(1)$.
\end{proof}

\begin{proof}\textbf{of Claim~\ref{claim_Delta2}.~}
Using above computation for $\mu(Q),\mu(Q')$, we have
\begin{align*}
\rho(Q,Q') \leq 7\sqrt{\frac{\ln 2n}{\binom{2n}{3}}}\left(\sqrt{p^3} + \sqrt{p^3+\gamma^3}\right)
\leq \frac{60p^{3/2}\sqrt{\ln n}}{n^{3/2}} 
\end{align*}
since $\gamma\leq p$ due to definition of $Q'_n$. 
On the other hand, using the fact that $\gamma \geq 5\sqrt{\frac{p\ln n}{n}}$, we have
\begin{align*}
| \mu(Q) - \mu(Q') | &= \left| \gamma^3 - \frac{3}{2n-1}\left(p^2\gamma + p\gamma^2\right)\right|
\\&\geq \gamma^3 \left(1 - \frac{3}{2n-1}\left(\frac{pn}{25\ln n} + \frac{\sqrt{np}}{5\sqrt{\ln n}}\right)\right)
\\&=\gamma^3(1- o_n(1)) \;, 
\end{align*}
which is at least $\frac{\gamma^3}{2}> \frac{60p^{3/2}\sqrt{\ln n}}{n^{3/2}}$ for large $n$. Hence, we have $Q'\in \overline{\Bc(Q,\rho)}$. 
\end{proof}

\subsection*{Proof of Corollary~\ref{cor_flambda_minimax}}
The proof follows the line of the previous proof, where we note that for $Q_n,Q'_n$ defined in Theorem~\ref{thm_separation}, we may write $\rho$ as
\begin{equation*}
\rho(Q_n,Q'_n) = \frac{t_n}{2n} \left(\sqrt{(2n-1)p_n} + \sqrt{(2n-1)p_n - \gamma_n}\right)
\end{equation*}
for some $t_n\in o_n(1)$.
We use the fact that $D_{Q_n} = (2n-1)p_n$ and $D_{Q'_n} = (2n-1)p_n - \gamma_n$,
and also observe that both distributions are in $\Fc_f$ if $p_n \geq \frac{(\ln 2n)^{4.1}}{2n}$.
We claim the following.
\begin{claim}
\label{claim_lambda1}
Let $p_n\in[\frac{(\ln 2n)^{4.1}}{2n}, \frac{1}{2}]$. There exists $\tau_n \in o_n(1)$ that satisfies both the following conditions: (i) $\gamma_n = \tau_n\sqrt{\frac{p_n(1-p_n)}{n}}$, and (ii) $Q'_n \in \overline{\Bc(Q_n,\rho)}$ where $\rho$ is given above.
\end{claim}
Hence, as in previous proof we have sequences $Q_n,Q'_n$ with vanishing total variation distance, and yet satisfying $H_1$. 
Now, we can use arguments similar to proof of Corollary~\ref{cor_fDelta_minimax} to arrive at the result. 
Note here, the restricting $p_n$ to be smaller than $\frac12$ is not a limitation in this case since the result follows as long as we can find some suitable sequence $p_n$.

As in previous subsection, we also give a proof of the fact that for larger $\gamma_n$, the proposed test can indeed distinguish between $Q_n$ and $Q'_n$. This result, stated below, applies only for $p_n\leq \frac12$.
\begin{claim} 
\label{claim_lambda2}
Let $p_n\in[\frac{(\ln 2n)^{4.1}}{2n}, \frac{1}{2}]$. If $\gamma_n \geq 15\sqrt{\frac{p_n}{n}}$, then $Q'_n \in \overline{\Bc(Q_n,\rho)}$ for $\rho$ mentioned in Corollary~\ref{cor_flambda_IER}. 
As a consequence, Corollary~\ref{cor_flambda_IER} implies that proposed test is consistent in distinguishing between $Q_n$ and $Q'_n$.
\end{claim}

\begin{proof}\textbf{of Claim~\ref{claim_lambda1}.~}
We drop the subscript $n$ for convenience. Note that for the $\rho$ defined above, we have $\rho(Q,Q') \leq t\sqrt{\frac{2p}{n}}$. Now, let $M_Q, M_{Q'}$ be the parameter matrix for $Q$ (ER), and $Q'$ (two community model), respectively, and recall that $\mu$ is the vector of the largest two singular values of the parameter matrix, scaled by the graph size $2n$.  
Hence, $\mu(Q) = \left( \frac{(2n-1)p}{2n}, \frac{p}{2n}\right)$.

On the other hand, $M_{Q'}$ has eigenvalues $((2n-1)p-\gamma)$ and $((2n-1)\gamma-p)$ each with multiplicity 1, and the remaining eigenvalues are $(-p-\gamma)$. 
Since $\gamma\leq p$, we have $((2n-1)p-\gamma) \geq ((2n-1)\gamma-p)$.
Also, if we let $\tau \geq 4\sqrt{\frac{2p}{n}}$, then one can verify that $\gamma\geq \tau\sqrt{\frac{p}{2n}}\geq \frac{2(p+\gamma)}{n}$, which implies that the largest two singular values of $M_{Q'}$ are first two eigenvalues.
Hence, we can write $\mu(Q') = \left( \frac{(2n-1)p}{2n} - \frac{\gamma}{2n}\,,\frac{(2n-1)\gamma}{2n} - \frac{p}{2n}\right)$.  
If $d$ is the Euclidean distance metric for $f_\lambda$, we have
\begin{align*}
d(\mu(Q),\mu(Q')) = \Vert \mu(Q) - \mu(Q')\Vert &= \sqrt{\left(\frac{\gamma}{2n}\right)^2 + \left(\gamma-\frac{2p+\gamma}{2n}\right)^2}
> \gamma-\frac{2p+\gamma}{2n} \geq \frac{\gamma}{2} \;,
\end{align*}
where the last inequality holds for $\tau \geq 4\sqrt{\frac{2p}{n}}$.
Combining this with the bound on $\rho$, we can conclude that $Q'\in\overline{\Bc(Q,\rho)}$ if $\tau>\max\left\{ 4t,4\sqrt{{2p}/{n}}\right\}$. Since, both terms are $o_n(1)$, we can easily choose a $\tau\in o_n(1)$ satisfying this condition.
\end{proof}

\begin{proof}\textbf{of Claim~\ref{claim_lambda2}.~}
Define $\rho(Q,Q') = \frac{7.5}{2n}\left({\sqrt{kD_Q}} + {\sqrt{kD_{Q'}}}\right)$, which is the lower bound for $\rho$ given by Corollary~\ref{cor_flambda_IER}. 
Using $k=2$ and the above computation for $\mu(Q),\mu(Q')$, we have
$\rho(Q,Q') \leq 7.5\sqrt{p/n}$.
On the other hand, the previously obtained lower bound $\Vert \mu(Q)-\mu(Q')\Vert \geq \frac{\gamma}{2}$ still holds.
Hence, $Q'\in\overline{\Bc(Q,\rho)}$ if $\gamma \geq 15\sqrt{p/n}$.
\end{proof}

\section{Discussion}
\label{sec_conclusion}
The main message of this paper is that two-sample testing is possible in the context of network comparison, where one may not have multiple observations from the same distribution. 
If one has access to a large population of networks generated from the same model, one may still use standard kernel based test statistics~\citep{Gretton_2012_jour_JMLR} in conjunction with graph kernels~\citep{Kondor_2016_conf_NIPS}.
However, a common situation in practice is where one has exactly two large networks, such as Facebook and LinkedIn connection networks or two brain networks, and needs to decide whether they are similar or different. 
The present paper concludes that it is indeed possible to address this problem statistically even when the graphs are defined on different entities.
However, a formal treatment requires certain considerations:
\begin{itemize}
\item 
There exists a network statistic $f$ that concentrates for large graphs (see Assumption~\ref{assum}). 
\item
If network statistic $f$ is used for comparison, the underlying testing problem merely compares between the point of concentration of $f$ for different models.
\item
It is often practical to ignore small separations between models, for example, if the comparison is between graphs of different sizes. We characterise this in terms of $\epsilon$.
\item
On the other hand, similar to the signal detection literature~\citep{Ingster_2000_jour_ESAIM,Baraud_2002_jour_Bernoulli}, one cannot always hope to distinguish between models with arbitrarily small separation (see Theorem~\ref{thm_separation} and subsequent corollaries).  
\end{itemize}
In addition, if there exists an accurate estimator $\widehat\sigma$ for the deviation of $f$ from its point of concentration, then we show that:
\begin{itemize}
\item
A general principle provides an uniformly consistent two-sample test for large graphs (Theorem~\ref{thm_main}).
The test does not require any knowledge of the underlying distribution class apart from function $\widehat\sigma$.
\item 
For specific network statistics, $f_\Delta$ and $f_\lambda$, and for $\IER$ graphs, the test is near-optimal in the minimax sense (see Sections~\ref{sec_examples} and~\ref{sec_minimax}).
\item
The test is also applicable for other network models (Corollary~\ref{cor_fDelta_Geom}), and even for comparing graphs generated from different distribution classes (Corollary~\ref{cor_fDelta_ERvGeom}).
\end{itemize}

Hence, we conclude that two-sample testing of large networks is possible whenever one has access to a network statistic that satisfies Assumption~\ref{assum}. In addition, if Assumption~\ref{assum_est} is satisfied, then one may also use the proposed two-sample test.
The discussion of this paper also leads to some interesting questions for further research. We state two important problems:
\begin{itemize}
\item
\textbf{Concentration of other popular network statistics}
\newline
In this paper, we have only considered triangle and spectrum based statistics as running examples.
The natural question one can ask is which other popularly used network statistics concentrate for generic model classes. 
For instance, concentration of functions like clustering coefficient and modularity in combination with our results will help to theoretically validate various claims about properties of brain networks. 
\item
\textbf{Bootstrapped variant of proposed two-sample test}
\newline
The proposed test primarily relies on concentration of the test statistic~\eqref{eqn_teststat} under the null and alternative hypotheses. 
It is often observed that the practical performance of concentration based tests can be improved by using bootstrapped variants~\citep{Gretton_2012_jour_JMLR,Tang_2016_jour_JCompGraphStat}.
In the present context, we feel that bootstrapping can help to achieve low error rate even for smaller and sparser graphs. 

However, bootstrapping is a challenging problem in the present setting.
\citet{Gretton_2012_jour_JMLR} consider a large population problem, where random mixing of the two population helps to estimate the null distribution for the test statistic. 
\citet{Tang_2016_jour_JCompGraphStat} deal with the two graph setting, but the assumption that the graphs are generated from RDPG model allows parameter estimation, which in turn, aids in generating bootstrapped samples from the estimated models.
It would be interesting to come up with bootstrapping procedures without such assumptions.  
\end{itemize}

\acks{The work of D. Ghoshdastidar and U. von Luxburg is supported by the German Research Foundation (Research Unit 1735) and the Institutional Strategy of the University of T{\"u}bingen (DFG, ZUK 63).
The work of M. Gutzeit and A. Carpentier is supported by the Deutsche Forschungsgemeinschaft (DFG) Emmy Noether grant MuSyAD (CA 1488/1-1).}

\bibliography{refs}

\end{document}